\documentclass[11pt]{amsart}

\usepackage{amsfonts}
\usepackage{amsmath,amssymb,amsthm}
\usepackage{enumerate,graphicx,xypic}
\usepackage[square,numbers]{natbib}

\newtheorem{theorem}{Theorem}[section]

\newtheorem{proposition}[theorem]{Proposition}

\theoremstyle{definition}

\newtheorem{algorithm}[theorem]{Algorithm}

\newtheorem{definition}[theorem]{Definition}
\newtheorem{example}[theorem]{Example}

\newtheorem{remark}[theorem]{Remark}

\begin{document}

\title{$k$-Mixing Properties of Multidimensional Cellular Automata}


\author{Chih-Hung Chang}
\address[Chih-Hung Chang]{Department of Applied Mathematics, National University of Kaohsiung, Kaohsiung 81148, Taiwan, ROC.}
\email{chchang@nuk.edu.tw}

\thanks{E-Mail: chchang@nuk.edu.tw. This work is partially supported by the Ministry of Science and Technology, ROC (Contract No MOST 104-2115-M-390-004-).}
\date{August 1, 2015}

\baselineskip=1.5\baselineskip

\begin{abstract}
This paper investigates the $k$-mixing property of a multidimensional cellular automaton. Suppose $F$ is a cellular automaton with the local rule $f$ defined on a $d$-dimensional convex hull $\mathcal{C}$ which is generated by an apex set $C$. Then $F$ is $k$-mixing with respect to the uniform Bernoulli measure for all positive integer $k$ if $f$ is a permutation at some apex in $C$. An algorithm called the \emph{Mixing Algorithm} is proposed to verify if a local rule $f$ is permutive at some apex in $C$. Moreover, the proposed conditions are optimal. An application of this investigation is to construct a multidimensional ergodic linear cellular automaton.
\end{abstract}

\maketitle

\section{Introduction}

Cellular automaton (CA) is a particular class of dynamical systems introduced by S. Ulam \cite{Ula-PICoM1952} and J. von Neumann \cite{vNeu-1966} as a model for self-production. CAs have been systematically studied by Hedlund from the viewpoint of symbolic dynamics \cite{Hed-MST1969}. Investigation of CAs from the point of view of the ergodic theory has received remarkable attention in the past few decades since CAs are widely applied in many disciplines such as biology, physics, computer science, and so on \cite{BKM-PD1997, CFMM-TCS2000, CFM+-TCS1999, FD-FI2008, Kle-PAMS1997, MM-BR2003, PY-ETDS2002}. Many dynamical behaviors of CAs are undecidable and the classification of dynamical behaviors is one of the central open questions in this field \cite{DFV-2003, KO-2008, LM-TCS2010, Lukkarila-JCA2010}.

Invertibility is one of the fundamental microscopic physical laws of nature. Bennett demonstrated that invertible Turing machines are computationally universal \cite{Bennett-IJRD1973}. The university remains true for one-dimensional cellular automata, even in the sense that any irreversible cellular automaton can be simulated by a reversible one on finite configurations \cite{Morit-JIPSJ1994,Morit-IPL1992,MH-IT1989,Morit-2012}. Amoroso and Patt showed that invertibility is decidable in one dimension \cite{AP-JCSS1972}, and Kari proved that it is undecidable in two and higher dimensions \cite{Kari-JCSS1994,Kari-TCS2005}. Ito \emph{et al.}\cite{ION-JCSS1983} present criteria for surjectivity and injectivity of the global transition map of one-dimensional linear CAs; they also mention that
\begin{quote}
\emph{criteria are desired for determining when the sequence of transitions of a state configuration of a cellular automata takes a certain type of dynamical behavior.}
\end{quote}
In this paper we propose a criterion for the dynamical behavior of multidimensional CAs in the framework of ergodic theory.

Shirvani and Rogers \cite{SR-CMP1991} show that a surjective CA with two symbols is invariant and strongly mixing with respect to uniform Bernoulli measure. Shereshevsky has studied some strong ergodic properties of the natural extension of a measure theoretic endomorphism such as $k$-mixing, and the number of symbols could be any positive integer \cite{She-MM1992}. One-dimensional surjective CAs admitting an equicontinuity point have a dense set of periodic orbits \cite{BT-AIHPPS2000}, and surjection and non-wandering are equivalent notions for multidimensional CAs \cite{ADF-IPL2013}.

Kleveland demonstrates that, for one-dimensional case, leftmost and rightmost permutive CAs are strongly mixing with respect to product measure defined by normalized Haar measure, and some bipermutive CAs are even $k$-mixing with respect to product measure \cite{Kle-PAMS1997}. Notably, leftmost and rightmost permutive CAs are both surjective \cite{Hed-MST1969}. Cattaneo \emph{et al}. propose an algorithm to construct ergodic $d$-dimensional linear CAs \cite{CFM+-TCS1999}. Some ergodic properties, such as ergodicity, strongly mixing, and Bernoulli automorphism, of one-dimensional CAs are revealed in \cite{Akin-GM2010, CA-2014}.

This investigation devotes to studying the surjection and $k$-mixing property of multidimensional cellular automata over a finite alphabet $\mathcal{A}$ with respect to the uniform Bernoulli measure $\mu$. Suppose $C \subset \mathbb{Z}^d$ is a finite subset in $d$-dimensional lattice. Let $\mathcal{C}$ be the convex hull in $\mathbb{Z}^d$ which is generated by $C$, and let $f$ be a map from $\mathcal{A}^{\mathcal{C}}$ to $\mathcal{A}$. Suppose $F$ is a CA with the local rule $f$ and $\mathcal{C}$ is a multidimensional hypercuboid. Proposition \ref{prop:corner-permutive-onto} indicates that $F$ is surjective if its local rule $f$ is \emph{corner permutive}, i.e., $f$ is a permutation at some vertex in $C$. This extends Hedlund's result \cite{Hed-MST1969} to multidimensional case.

Furthermore, Theorem \ref{thm:rectangle-k-mixing} demonstrates that $F$ is $k$-mixing with respect to the uniform Bernoulli measure for all $k \geq 1$ if $f$ defined on a hypercuboid is corner permutive. Note that, for the case $k = 1$, $F$ is known as strongly mixing. Theorem \ref{thm:polygon-k-mixing} extends Theorem \ref{thm:rectangle-k-mixing} to more general case.

Suppose $\mathcal{C}$ is a convex hull generated by an apex set $C$ and $\mathcal{C}$ is not a hypercuboid. Theorem \ref{thm:polygon-k-mixing} addresses an algorithm named \emph{Mixing Algorithm} to verify if $F$ is $k$-mixing with respect to the uniform Bernoulli measure. Roughly speaking, $F$ is $k$-mixing for $k \geq 1$ if the local rule $f$ is permutive at some apex in $C$. It is remarkable that Theorems \ref{thm:rectangle-k-mixing} and \ref{thm:polygon-k-mixing} can be extended to any Markov measure $\nu$ as long as it is $F$-invariant.

In \cite{Willson-MST1975}, Wilson demonstrate that a two-dimensional linear CA over $\mathcal{A} = \{0, 1\}$ is mixing if it local rule is permutive in some extremal coordinate $x$, where $x$ is called \emph{extremal} if $\langle x, x\rangle <\langle x, y\rangle$ for all $y \in \mathcal{C}$. Wilson's proof technique can easily generalize to any extremally permutive CAs on any alphabet, and any $d$-dimensional linear CA for $d \geq 2$. In \cite{Lee-2009}, Lee reveals Wilson's result holds for two-dimensional CA if its local rule is permutive at the corner. Theorem \ref{thm:polygon-k-mixing} generalizes Wilson's and Lee's result to more general case. For instance, Examples \ref{eg:2d-mixing-polygon-case2} and \ref{eg:3d-mixing-polygon} are both $k$-mixing for all $k \in \mathbb{N}$, and neither of them is extremally permutive.

It is worth emphasizing that the conditions proposed in Theorems \ref{thm:rectangle-k-mixing} and \ref{thm:polygon-k-mixing} are optimized already. Example \ref{eg:not-corner-permutive-not-mixing} provides an two-dimensional instance which illustrates a CA with non-corner-permutive local rule being not even ergodic. An application of the present investigation is the construction of multidimensional ergodic linear CAs, which is different from Cattaneo \emph{et al}.\cite{CFM+-TCS1999} and is elucidated in the further work.

The rest of the paper is organized as follows. Section 2 establishes some basic definitions and formulations of problem to state the main theorems. The condition that determines whether or not a multidimensional CA is surjective is addressed therein. Sections 3 and 4 deliberate the $k$-mixing property of a multidimensional CA with the local rule defined on a hypercuboid and a convex hull, respectively. An example infers that the conditions proposed in Theorems \ref{thm:rectangle-k-mixing} and \ref{thm:polygon-k-mixing} are optimal and some discussion are stated in Section 5.

\section{Preliminary}

Let $\mathcal{A}=\{0,1,\cdots,m-1\}$ be a finite alphabet for some positive integer $m>1$ and let $\mathbf{X}=\mathcal{A}^{\mathbb{Z}^d}$ be the $d$-dimensional lattice over $\mathcal{A}$. Namely,
$$
\mathbf{X} = \{x=(x_{\mathbf{i}})_{\mathbf{i} \in {\mathbb{Z}^d}}: x_{\mathbf{i}} \in \mathcal{A}\}.
$$
A $d$-dimensional cellular automaton (CA) $F$ is defined as follows. Suppose $D$ is a finite subset of $\mathbb{Z}^d$ and $f: \mathcal{A}^D \to \mathcal{A}$ is given as a local map, where
$$
\mathcal{A}^D = \{x=(x_{\mathbf{i}})_{\mathbf{i} \in D}: x_{\mathbf{i}} \in \mathcal{A}\}.
$$
The map $F: \mathbf{X} \to \mathbf{X}$ defined by $F(x)_{\mathbf{i}} = f(x_{\mathbf{i} + D})$, herein $\mathbf{i} + D = \{\mathbf{i} + \mathbf{j}: \mathbf{j} \in D\}$ and $x_{\mathbf{i} + D} = (x'_{\mathbf{k}})_{\mathbf{k} \in D}$ with $x'_{\mathbf{k}} = x_{\mathbf{i} + \mathbf{k}}$, is called the CA  with the local rule $f$. A CA $F$ with the local rule $f$ is called linear if $f$ is linear, i.e., $f(x) = \Sigma_{\mathbf{i} \in D} a_{\mathbf{i}} x_{\mathbf{i}} \pmod m$.

For instance, suppose $t, b, l, r \in \mathbb{Z}$ satisfy $b \leq t, l \leq r$. Let $C = \{(l, b), (l, t), (r, b), (r, t)\}$ be a finite subset in $\mathbb{Z}^2$, and let $\mathcal{C} \subset \mathbb{Z}^2$ be the polygon generated by $C$. In other words,
$$
\mathcal{C} = \textrm{poly}(C) = \{(i, j) \in \mathbb{Z}^2: l \leq i \leq r, b \leq j \leq t\}
$$
is a rectangle. Define $f: \mathcal{A}^{\mathcal{C}} \to \mathcal{A}$ by
$$
f \left(
    \begin{array}{ccc}
    x_{l, t} &\cdots & x_{r, t} \\
    \vdots &\ddots &\vdots  \\
    x_{l, b} &\cdots & x_{r, b} \\
    \end{array} \right)
= \sum_{l \leq i \leq r, b \leq j \leq t} a_{i, j} x_{i, j} \pmod m,
$$
where $a_{i, j} \in \mathbb{Z}$ for $l \leq i \leq r, b \leq j \leq t$. Then the CA $F$ with the local rule $f$ is given by
$$
F(x)_{i, j} = f \left(
    \begin{array}{ccc}
    x_{i+l, j+t} &\cdots & x_{i+r, j+t} \\
    \vdots &\ddots &\vdots  \\
    x_{i+l, j+b} &\cdots & x_{i+r, j+b} \\
    \end{array} \right)
$$
and is a linear CA.

The study of the local rule of a CA is essential for the understanding of this system. In \cite{Hed-MST1969}, Hedlund introduced a terminology \emph{permutive} for one-dimensional CAs. The following definition extends Hedlund's definition to multidimensional case.

\begin{definition}
A local rule $f: \mathcal{A}^D \to \mathcal{A}$ is called permutive in the variable $x_{\mathbf{i}}$, $\mathbf{i} \in D$, if $f$ is a permutation at $x_{\mathbf{i}}$. More precisely, for each set
$\{c_{\mathbf{j}}: \mathbf{j} \in D \setminus \{\mathbf{i}\}, c_{\mathbf{j}} \in \mathcal{A}\}$, the map $g: \mathcal{A} \to \mathcal{A}$ defined by $g(a) = f(x^a)$ one-to-one and onto, where
$$
x^a_{\mathbf{j}} = \left\{
\begin{array}{ll}
a, & \mathbf{j} = \mathbf{i}; \\
c_{\mathbf{j}}, & \hbox{otherwise.}
\end{array}\right.
$$
\end{definition}

A straightforward verification demonstrates the following proposition.

\begin{proposition}
Suppose $f: \mathcal{A}^D \to \mathcal{A}$ is linear. Then $f(x) = \Sigma_{\mathbf{i} \in D} a_{\mathbf{i}} x_{\mathbf{i}}$ is permutive in the variable $x_{\mathbf{j}}$ if and only if $a_{\mathbf{j}}$ is relative prime to $m$.
\end{proposition}

Suppose $F$ is a one-dimensional CA with the local rule $f = f(x_i, \ldots, x_j)$, where $i \leq j$. $f$ is called leftmost permutive (respectively rightmost permutive) if $f$ is permutive in the variable $x_i$ (respectively $x_j$). Then $F$ is surjective if $f$ is either leftmost permutive or rightmost permutive (\cite{Hed-MST1969}). A careful and routine examination extends Hedlund's result to multidimensional CAs.

For $1 \leq i \leq d$, suppose $k_i, K_i \in \mathbb{Z}$ are given so that $k_i \leq K_i$. Let $C$ be the subset of $\mathbb{Z}^d$ such that every coordinate of element in $C$ is either $k_i$ or $K_i$, and let $\mathcal{C}$ be the convex hull generated by $C$. It is seen that $\mathcal{C}$ is a $d$-dimensional hypercuboid. A local map $f$ defined on $\mathcal{C}$ is called \emph{corner permutive} if $f$ is permutive in the variable $x_{\mathbf{i}}$ for some $\mathbf{i} \in C$. Proposition \ref{prop:corner-permutive-onto} addresses a sufficient condition for the discrimination of surjection of a CA.

\begin{proposition}\label{prop:corner-permutive-onto}
If $f: \mathcal{A}^{\mathcal{C}} \to \mathcal{A}$ is corner permutive, where $\mathcal{C}$ is a hypercuboid. Then the CA $F$ defined by the local rule $f$ is surjective.
\end{proposition}

An immediate application of Proposition \ref{prop:corner-permutive-onto} is that a linear CA with the local rule $f(x_{\mathcal{C}}) = \Sigma_{\mathbf{i} \in \mathcal{C}} a_{\mathbf{i}} x_{\mathbf{i}}$ is surjective if $\textrm{gcd}(a_{\mathbf{i}}, m) = 1$ for some $\mathbf{i}$ with $\mathbf{i}_j \in \{k_j, K_j\}$ for $1 \leq j \leq d$, where $\textrm{gcd}(p, q)$ means the greatest common divisor of $p$ and $q$.

Let $(X, \mathcal{B}, \mu)$ be a probability space and let $T: X \to X$ be measure-preserving transformation, i.e., $\mu(T^{-1} A) = \mu(A)$ for $A \in \mathcal{B}$. $T$ is called \emph{ergodic} if every measurable subset $A \subseteq X$ with $T^{-1} A = A$ satisfies $\mu(A) = 0$ or $\mu(A) = 1$. The following theorem addresses some equivalent conditions for the ergodicity of $T$.

\begin{theorem}[See \cite{Wal-1982}]
Suppose $T: X \to X$ is a measure-preserving transformation on a probability space $(X, \mathcal{B}, \mu)$. The following statements are equivalent.
\begin{enumerate}[\bf (i)]
\item $T$ is ergodic.
\item For every $A, B \in \mathcal{B}$ with $\mu(A) > 0, \mu(B) > 0$, there exists $n \in \mathbb{N}$ such that $\mu(A \cap T^{-n} B) > 0$.
\end{enumerate}
\end{theorem}

A stronger property for a measure-preserving transformation is \emph{strongly mixing}.

\begin{definition}
Let $(X, \mathcal{B}, \mu)$ be a probability space and let $T: X \to X$ be measure-preserving transformation.
\begin{enumerate}[\bf (i)]
\item $T$ is called \emph{strongly mixing} if
$$
\lim_{n \to \infty} \mu(A \cap T^{-n} B) = \mu(A) \mu(B)
$$
for every $A, B \in \mathcal{B}$.
\item $T$ is called \emph{$k$-mixing} if
$$
\lim_{n_1, \ldots, n_k \to \infty} \mu(A_0 \cap T^{-n_1} A_1 \cap \cdots \cap T^{-n_k} A_k) = \mu(A_0) \mu(A_1) \cdots \mu(A_k)
$$
for every $\{A_i\}_{i=0}^k \in \mathcal{B}$.
\end{enumerate}
\end{definition}

It comes immediately that, for a measure-preserving transformation $T$:
\begin{center}
$k$-mixing \quad $\Rightarrow$ \quad strongly mixing \quad $\Rightarrow$ \quad ergodic
\end{center}

\section{Mixing Property for Local Rules on Hypercuboid}

This section studies the mixing properties of multidimensional cellular automata with the local rules defined on the hypercuboid. Let $\pi_j: \mathbb{R}^d \rightarrow \mathbb{R}$ be the canonical projection on the $j$th coordinate, i.e., $\pi_j(v) = v_j$, where $v = (v_1, \ldots, v_d) \in \mathbb{R}^d$ and $1 \leq j \leq d$.

\begin{theorem}\label{thm:rectangle-k-mixing}
Suppose $F$ is a linear CA with the local rule $f(x_{\mathcal{C}}) = \Sigma_{\mathbf{i} \in \mathcal{C}} a_{\mathbf{i}} x_{\mathcal{C}}$ defined on a $d$-dimensional hypercuboid $\mathcal{C}$, where $\mathcal{C}$ is the convex hull generated by the set
$$
C = \{\mathbf{i} = (i_1, \ldots, i_d): i_j \in \{k_j, K_j\}, 1 \leq j \leq d\}
$$
and $\{k_j, K_j\}_{j=1}^d \subset \mathbb{Z}$. If $\textrm{gcd}(a_{\mathbf{i}}, m) = 1$ for some $\mathbf{i}$ such that
\begin{equation}\label{cond:rectangle-mixing}
i_j \left\{
\begin{array}{ll}
> 0, & i_j = K_j; \\
< 0, & i_j = k_j.
\end{array}
\right.
\end{equation}
then $F$ is $k$-mixing with respect to the uniform Bernoulli measure $\mu$ for $k \geq 1$.
\end{theorem}
\begin{proof}
Denote a $d$-dimensional cylinder $\mathbf{C}$ by
\begin{align*}
\mathbf{C} &= <(\mathbf{v}_1, c_1), (\mathbf{v}_2, c_2), \ldots, (\mathbf{v}_n, c_n)> \\
  &= \{x \in \mathcal{A}^{\mathbb{Z}^d}: x_{\mathbf{v}_i} = c_i, 1 \leq i \leq n\},
\end{align*}
where $\mathbf{v}_i \in \mathbb{Z}^d$ and $c_i \in \mathcal{A}$ for all $i$.

The proof of Theorem \ref{thm:rectangle-k-mixing} is divided into several steps. In addition, demonstration of Theorem \ref{thm:rectangle-k-mixing} for the case $d = 2$ is addressed to clarify the procedures. Proof for the general case is analogous, thus is omitted.

For the case that there exist $l, r, b, t \in \mathbb{Z}$ with $l \leq r$ and $b \leq t$ such that
$$
\{\mathbf{v}_i\}_{i=1}^n = \{(j_1, j_2): l \leq j_1 \leq r, b \leq j_2 \leq t\}
$$
the cylinder $\mathbf{C}$ is then denoted as
$$
\mathbf{C} = \left<
\begin{array}{lll}
  c_{(l, t)} & \cdots & c_{(r, t)} \\
  \vdots &\ddots &\vdots  \\
  c_{(l, b)} & \cdots & c_{(r, b)}
\end{array}
\right>^{(r, t)}_{{\hspace{-10em}(l, b)}}.
$$
Notably, in this case,
$$
\{c_{\mathbf{v}_i}\}_{i=1}^n = \{c_{(j_1, j_2)}: l \leq j_1 \leq r, b \leq j_2 \leq t\}.
$$

When $d = 2$, the local rule $f: \mathcal{A}^{(r-l+1) \times (t-b+1)} \to \mathcal{A}$ is defined on a rectangle and is expressed as
$$
f \left(
\begin{array}{lll}
  x_{l, t} & \cdots & x_{r, t} \\
  \vdots &\ddots &\vdots  \\
  x_{l, b} & \cdots & x_{r, b}
\end{array}
\right)
= \sum_{l \leq j_1 \leq r, b \leq j_2 \leq t} a_{j_1, j_2} x_{j_1, j_2} \pmod m.
$$
Then $F$ is $k$-mixing with respect to the uniform Bernoulli measure for $k \geq 1$ if its local rule $f$ satisfies either one of the following conditions.
\begin{enumerate}[\bf (i)]
  \item $\textrm{gcd}(a_{r, t}, m) = 1$ and $r, t > 0$.
  \item $\textrm{gcd}(a_{r, b}, m) = 1$ and $r > 0, b < 0$.
  \item $\textrm{gcd}(a_{l, b}, m) = 1$ and $l, b < 0$.
  \item $\textrm{gcd}(a_{l, t}, m) = 1$ and $l < 0, t > 0$.
\end{enumerate}

\noindent \textbf{Step 1.} Let $\mathcal{L}$ be the collection of linear local rules and let
$$
\mathbb{Z}_m[x, x^{-1}, y, y^{-1}] = \{\sum_{p_1 \leq i \leq p_2, q_1 \leq j \leq q_2} a_{i, j} x^i y^j, p_1, p_2, q_1, q_2 \in \mathbb{Z}, a_{i, j} \in \mathbb{Z}_m\}
$$
where $\mathbb{Z}_m = \{0, 1, \ldots, m-1\}$ is the ring of the integers modulo $m$. Define $\chi: \mathcal{L} \to \mathbb{Z}_m[x, x^{-1}]$ as
$$
\chi (\sum_{l \leq j_1 \leq r, b \leq j_2 \leq t} a_{j_1, j_2} x_{j_1, j_2}) = \sum_{l \leq j_1 \leq r, b \leq j_2 \leq t} a_{j_1, j_2} x^{-j_1} y^{-j_2}.
$$
It follows that $\chi$ is a bijective map. Moreover, let $\mathbb{Z}_m[[x, x^{-1}, y, y^{-1}]]$ denote the ring of formal power series generated by $\{x, x^{-1}, y, y^{-1}\}$ over $\mathbb{Z}_m$. Define $\widehat{\chi}: \mathbf{X} = \mathcal{A}^{\mathbb{Z}^2} \to \mathbb{Z}_m[[x, x^{-1}, y, y^{-1}]]$ as
$$
\widehat{\chi}(\mathbf{b})=\sum^\infty_{j_1, j_2=-\infty} b_{j_1, j_2} x^{j_1} y^{j_2}, \quad \text{where} \quad \mathbf{b} = (b_{j_1, j_2})_{j_1, j_2 \in \mathbb{Z}} \in \mathbf{X}.
$$
A straightforward verification indicates that $\widehat{\chi}$ is one-to-one and onto. Observe that, for each $\mathbf{b} = (b_{j_1, j_2}) \in \mathbf{X}$,
\begin{align*}
\widehat{\chi}(F(\mathbf{b})) &= \widehat{\chi}\left[\left(\sum_{l \leq p \leq r, b \leq q \leq t} a_{p, q} b_{j_1+p, j_2+q}\right)_{j_1, j_2 \in \mathbb{Z}}\right] \\
 &= \sum^\infty_{j_1, j_2=-\infty} \left(\sum_{l \leq p \leq r, b \leq q \leq t} a_{p, q} b_{j_1+p, j_2+q}\right) x^{j_1} y^{j_2}.
\end{align*}
Furthermore, let $\mathbb{F}: \mathbb{Z}_m[[x, x^{-1}, y, y^{-1}]] \to \mathbb{Z}_m[[x, x^{-1}, y, y^{-1}]]$ be given by
$$
\mathbb{F}(P(x, x^{-1}, y, y^{-1})) = \chi (f) \cdot P(x, x^{-1}, y, y^{-1}).
$$
The equality
\begin{align*}
\mathbb{F}(\widehat{\chi}(\mathbf{b})) &= \mathbb{F} \left(\sum^\infty_{j_1, j_2=-\infty} b_{j_1, j_2} x^{j_1} y^{j_2}\right) \\
 &= \left(\sum_{l \leq p \leq r, b \leq q \leq t} a_{p, q} x^{-p} y^{-q}\right) \left(\sum^\infty_{j_1, j_2=-\infty} b_{j_1, j_2} x^{j_1} y^{j_2}\right) \\
 &= \sum^\infty_{j_1, j_2=-\infty} \left(\sum_{l \leq p \leq r, b \leq q \leq t} a_{p, q} b_{j_1+p, j_2+q}\right) x^{j_1} y^{j_2}
\end{align*}
together with the above equation demonstrate that the diagram
$$
\xymatrix{
\mathbf{X} = \mathcal{A}^{\mathbb{Z}^2} \ar[rr]^F \ar[d]_{\widehat{\chi}} && \mathbf{X} = \mathcal{A}^{\mathbb{Z}^2} \ar[d]^{\widehat{\chi}} \\
\mathbb{Z}_m[[x,x^{-1}, y, y^{-1}]] \ar[rr]_{\mathbb{F}} && \mathbb{Z}_m[[x,x^{-1}, y, y^{-1}]]}
$$
commutes.

\noindent \textbf{Step 2.} Without ambiguousness we abuse the notation $T^n = T \circ T^{n-1}$ to indicate the $n$th iteration of $T$, and abuse $\mathbb{F}^n$ to means the $n$th power of $(\chi (f))$. Namely, $\mathbb{F}^n = (\chi (f))^n$. It comes from the definitions that $F^n$ is a CA with the local rule $f^n$.

Notably, the mathematical induction infers $f^n = \chi^{-1}(\mathbb{F}^n)$ for all positive integer $n$. It is seen that $\mathrm{gcd}(a_{r, t}, m) = 1$ is followed by $\mathrm{gcd}(a_{r, t}^n, m) = 1$ for $n \in \mathbb{N}$. Combining above facts with $\mathbb{F}^n = a_{r, t}^n x^{-nr} y^{-nt} + \Sigma_{(j_1, j_2) \lneqq (nr, nt)} c_{j_1, j_2} x^{-j_1} y^{-j_2}$ for some $c_{j_1, j_2}$ demonstrate that $f$ is permutive in the variable $x_{r, t}$ leads to $f^n$ is permutive in the variable $x_{nr, nt}$. More specifically, if $f$ is corner permutive, then $f^n$ remains corner permutive in the same direction.

\noindent \textbf{Step 3.} Suppose $\mathbf{C}_0 = <(\mathbf{v}_1^0, c_1^0), \ldots, (\mathbf{v}_{l_0}^0, c_{l_0}^0)>$ and $\mathbf{C}_1 = <(\mathbf{v}_1^1, c_1^1), \ldots, (\mathbf{v}_{l_1}^1, c_{l_1}^1)>$ are two cylinders in $\mathbf{X}$, and $\mathrm{gcd}(a_{r, t}^n, m) = 1$ with $r, t \in \mathbb{N}$. For each finite subset $C \subset \mathbb{Z}^2$, define $M_i, m_i: C \to \mathbb{Z}$ as
\begin{align*}
M_i (C) &= \max\{v_i: v = (v_1, v_2) \in C\}, \\
m_i (C) &= \min\{v_i: v = (v_1, v_2) \in C\}
\end{align*}
for $i = 1, 2$. Set
$$
n_0 = \max \left\{ \left\lceil \dfrac{M_1(\mathbf{C}_0) - m_1(\mathbf{C}_1)}{r} \right\rceil, \left\lceil \dfrac{M_2(\mathbf{C}_0) - m_2(\mathbf{C}_1)}{t} \right\rceil \right\},
$$
where $\lceil k \rceil$ denotes the smallest integer that is greater than or equal to $k$. It comes immediately that $m_1(\mathbf{C}_1) + nr > M_1(\mathbf{C}_0)$ and $m_2(\mathbf{C}_1) + nt > M_2(\mathbf{C}_0)$ for all $n > n_0$.

Step 2 illustrates that $F^n$ is a CA with the local rule $f^n$ which is a permutation at $x_{nr, nt}$. Fix $i = 1, 2, \ldots, l_0$, for each given
$$
\{x_{\mathbf{v}^1_{i,1} + p, \mathbf{v}^1_{i,2} + q}\}_{nl \leq p \leq nr, nb \leq q \leq nt, (p, q) \neq (nr, nt)}, \text{ where } \mathbf{v}^1_i = (\mathbf{v}^1_{i,1}, \mathbf{v}^1_{i,2}),
$$
there is a unique $x_{\mathbf{v}_{i,1}^1+nr, \mathbf{v}_{i,1}^1+nt} \in \mathcal{A}$ such that
$$
a^1_i =f^n\left(
      \begin{array}{ccc}
      x_{\mathbf{v}_{i,1}^1+nl, \mathbf{v}_{i,2}^1+nt} &\cdots & x_{\mathbf{v}_{i,1}^1+nr, \mathbf{v}_{i,1}^1+nt} \\
      \vdots &\ddots &\vdots  \\
      x_{\mathbf{v}_{i,1}^1+nl, \mathbf{v}_{i,2}^1+nb} &\cdots & x_{\mathbf{v}_{i,1}^1+nr, \mathbf{v}_{i,2}^1+nt} \\
      \end{array}
    \right).
$$
Furthermore, $F^{-n}(\mathbf{C}_1) = \bigcap_{i=1}^{l_1} F^{-n}(<(\mathbf{v}^1_i, a^1_i)>)$, where
\begin{align*}
&F^{-n}(<(\mathbf{v}^1_i, a^1_i)>) = \left\{ \left<
      \begin{array}{ccc}
      x_{\mathbf{v}_{i,1}^1+nl, \mathbf{v}_{i,2}^1+nt} &\cdots & x_{\mathbf{v}_{i,1}^1+nr, \mathbf{v}_{i,1}^1+nt} \\
      \vdots &\ddots &\vdots  \\
      x_{\mathbf{v}_{i,1}^1+nl, \mathbf{v}_{i,2}^1+nb} &\cdots & x_{\mathbf{v}_{i,1}^1+nr, \mathbf{v}_{i,2}^1+nt} \\
      \end{array}
    \right>:\right. \\
  &\quad \left. f^n\left(
      \begin{array}{ccc}
      x_{\mathbf{v}_{i,1}^1+nl, \mathbf{v}_{i,2}^1+nt} &\cdots & x_{\mathbf{v}_{i,1}^1+nr, \mathbf{v}_{i,1}^1+nt} \\
      \vdots &\ddots &\vdots  \\
      x_{\mathbf{v}_{i,1}^1+nl, \mathbf{v}_{i,2}^1+nb} &\cdots & x_{\mathbf{v}_{i,1}^1+nr, \mathbf{v}_{i,2}^1+nt} \\
      \end{array}
    \right) = a_i^1, 1 \leq i \leq l_1 \right\}.
\end{align*}

\noindent \textbf{Step 4.} To see that $\mu(\mathbf{C}_0 \bigcap F^{-n}\mathbf{C}_1) = \mu(\mathbf{C}_0) \mu(\mathbf{C}_1)$, the discussion relies on the cases that $l$ and $b$ are positive/negative/zero. The case that $l, b < 0$ is addressed herein, the other cases can be elucidated analogously. Since $l$ and $b$ are both negative, it is seen that
\begin{equation}\label{ineq:rectangle-mixing1}
m_1(\mathbf{C}_1) + nl < m_1(\mathbf{C}_0) \leq M_1(\mathbf{C}_0) < m_1(\mathbf{C}_1)+nr < M_1(\mathbf{C}_1) + nr,
\end{equation}
and
\begin{equation}\label{ineq:rectangle-mixing2}
m_2(\mathbf{C}_1) + nb < m_2(\mathbf{C}_0) \leq M_2(\mathbf{C}_0) < m_2(\mathbf{C}_1) + nt < M_2(\mathbf{C}_1) + nt.
\end{equation}
Notably, for $1 \leq i \leq l_1$, the cardinality of $F^{-n}(<(\mathbf{v}^1_i, a^1_i)>)$ is $(nr - nl+1)(nt-nb+1)-1$. Equations \eqref{ineq:rectangle-mixing1} and \eqref{ineq:rectangle-mixing2} infers that the coordinates of $\mathbf{C}_0$ are covered by the coordinates of $F^{-n}(<(\mathbf{v}^1_i, a^1_i)>)$, more precisely,
$$
\{\mathbf{v}_i^0\}_{i=1}^{l_0} \subsetneq \{(\mathbf{v}^1_{i,1} + p, \mathbf{v}^1_{i,2} + q)\}_{nl \leq p \leq nr, nb \leq q \leq nt, (p, q) \neq (nr, nt)}.
$$
Therefore,
$$
\mu(\mathbf{C}_0 \bigcap F^{-n}(<(\mathbf{v}^1_1, a^1_1)>)) = \left(\dfrac{1}{m}\right)^{l_0} \cdot \dfrac{1}{m}.
$$
Similar discussion reveals that
$$
\mu(\mathbf{C}_0 \bigcap F^{-n}(<(\mathbf{v}^1_1, a^1_1)>) \bigcap F^{-n}(<(\mathbf{v}^1_2, a^1_2)>)) = \left(\dfrac{1}{m}\right)^{l_0} \cdot \left(\dfrac{1}{m}\right)^2.
$$
Repeating the procedures demonstrates
$$
\mu(\mathbf{C}_0 \bigcap F^{-n} \mathbf{C}_1) = \left(\dfrac{1}{m}\right)^{l_0} \cdot \left(\dfrac{1}{m}\right)^{l_1} = \mu(\mathbf{C}_0) \mu(\mathbf{C}_1)
$$
for $n > n_0$.

\noindent \textbf{Step 5.} For a fixed positive integer $k$, to prove that $F$ is $k$-mixing, it suffices to show
$$
\mu(\mathbf{C}_0 \bigcap F^{-n_1} \mathbf{C}_1 \bigcap \cdots \bigcap F^{-(n_1+...+n_k)} \mathbf{C}_k)=\mu(\mathbf{C}_0)\mu(\mathbf{C}_1)\cdots\mu(\mathbf{C}_k)
$$
for any cylinders $\mathbf{C}_0, \mathbf{C}_1, \cdots, \mathbf{C}_k \subset \mathbf{X}$ and $n_1, n_2, \ldots, n_k \in \mathbb{N}$ large enough. Suppose these cylinders are given as
$$
\mathbf{C}_i=<(\mathbf{v}_1^i, a_1^i), (\mathbf{v}_2^i, a_2^i), \cdots, (\mathbf{v}_{l_i}^i, a_{l_i}^i)>, \quad i = 0, 1, \ldots, k.
$$
Set
$$
n_0 = \max_{1 \leq i \leq k} \left\{ \left\lceil \dfrac{M_1(\mathbf{C}_{i-1})-m_1(\mathbf{C}_i)}{r} \right\rceil, \left\lceil \dfrac{M_2(\mathbf{C}_{i-1})-m_2(\mathbf{C}_i)}{t} \right\rceil \right\}.
$$
Pick $n_1, n_2, \ldots, n_k \geq n_0$, and let
$$
N_0 = 0, N_i=\sum_{j=1}^{i} n_j \text{ for } 1 \leq i \leq k.
$$
Define
\begin{align*}
i_- &=\min\{m_1(\mathbf{C}_i)+lN_i : 0 \leq i \leq k\}, & i_+ &= M_1(\mathbf{C}_k)+rN_k, \\
j_- &=\min\{m_2(\mathbf{C}_i)+bN_i : 0 \leq i \leq k\}, & j_+ &= M_2(\mathbf{C}_k)+tN_k.
\end{align*}
Similar elucidation to the discussion in Steps $2 \sim 4$ reveals that the set $\bigcap_{i=0}^k F^{-N_i} \mathbf{C}_i$ is the intersection of cylinders of the form
$$
\mathbf{C} = \quad
    \left\langle
      \begin{array}{ccc}
      a_{(i_-,j_+)} &\cdots & a_{(i_+,j_+)} \\
      \vdots &\ddots &\vdots  \\
      a_{(i_-,j_-)} &\cdots & a_{(i_+,j_-)} \\
      \end{array}
    \right\rangle^{(i_+,j_+)}_{\hspace{-13.2em}(i_-,j_-)}
$$
and the coordinates of $F^{-N_{i-1}} \mathbf{C}_{i-1}$ are covered by the coordinates of $F^{-N_i} \mathbf{C}_i$ for $i = 1, 2, \ldots, k$. This leads to the desired equality
$$
\mu \left(\bigcap_{i=0}^k F^{-N_i} \mathbf{C}_i \right) = \mu(\mathbf{C}_0) \mu(\mathbf{C}_1) \cdots\mu(\mathbf{C}_k).
$$
Namely, $F$ is $k$-mixing with respect to the uniformly Bernoulli measure $\mu$ for $k \geq 1$.

The other cases can be done analogously, this completes the proof.
\end{proof}

\begin{remark} \label{rmk:rectangle-k-mixing-permutive}
It is remarkable that Theorem \ref{thm:rectangle-k-mixing} remains true for any CA whose local rule is permutive in the variable $x_{\mathbf{i}}$ satisfying \eqref{cond:rectangle-mixing}.
\end{remark}

\begin{example}\label{eg:2d-self-replication}
Self-replicating pattern generation is an interesting topic in nonlinear science. A \emph{motif} is considered as a basic pattern. Self-replicating pattern generation is the process of transforming copies of the motif about the space in order to create the whole repeating pattern with no overlaps and blanks. The following proposes a cellular automaton which is capable of self-replication. Moreover, Theorem \ref{thm:rectangle-k-mixing} reveals its dynamical behavior.

Suppose $F$ is a two-dimensional CA over alphabet $\mathcal{A} = \{0, 1\}$ whose local rule $f: \mathcal{A}^{\mathcal{C}} \to \mathcal{A}$ is given by
$$
f (x_{\mathcal{C}}) = x_{0, 0} + x_{0, 1} + x_{1, 1} \pmod 2,
$$
where $\mathcal{C} = \{(0, 0), (1, 0), (0, 1), (1, 1)\}$. Figure \ref{fig:2d-self-replication} indicates that $F$ self-replicate initial patterns at the $2^n$th step for $n \geq 4$. More precisely, three copies of the initial patterns are reproduced, as we can see, at the 16th time step (center figure) and the 32nd time step (right bottom figure). Meanwhile, some interesting patterns are observed in the transformation. The numerical experiment is carried out in a $100 \times 100$ grids with periodic boundary condition.

Moreover, Proposition \ref{prop:corner-permutive-onto} and Theorem \ref{thm:rectangle-k-mixing} demonstrate that $F$ is surjective and $k$-mixing for $k \in \mathbb{N}$.

\begin{figure}
\begin{center}
\includegraphics[scale=0.3]{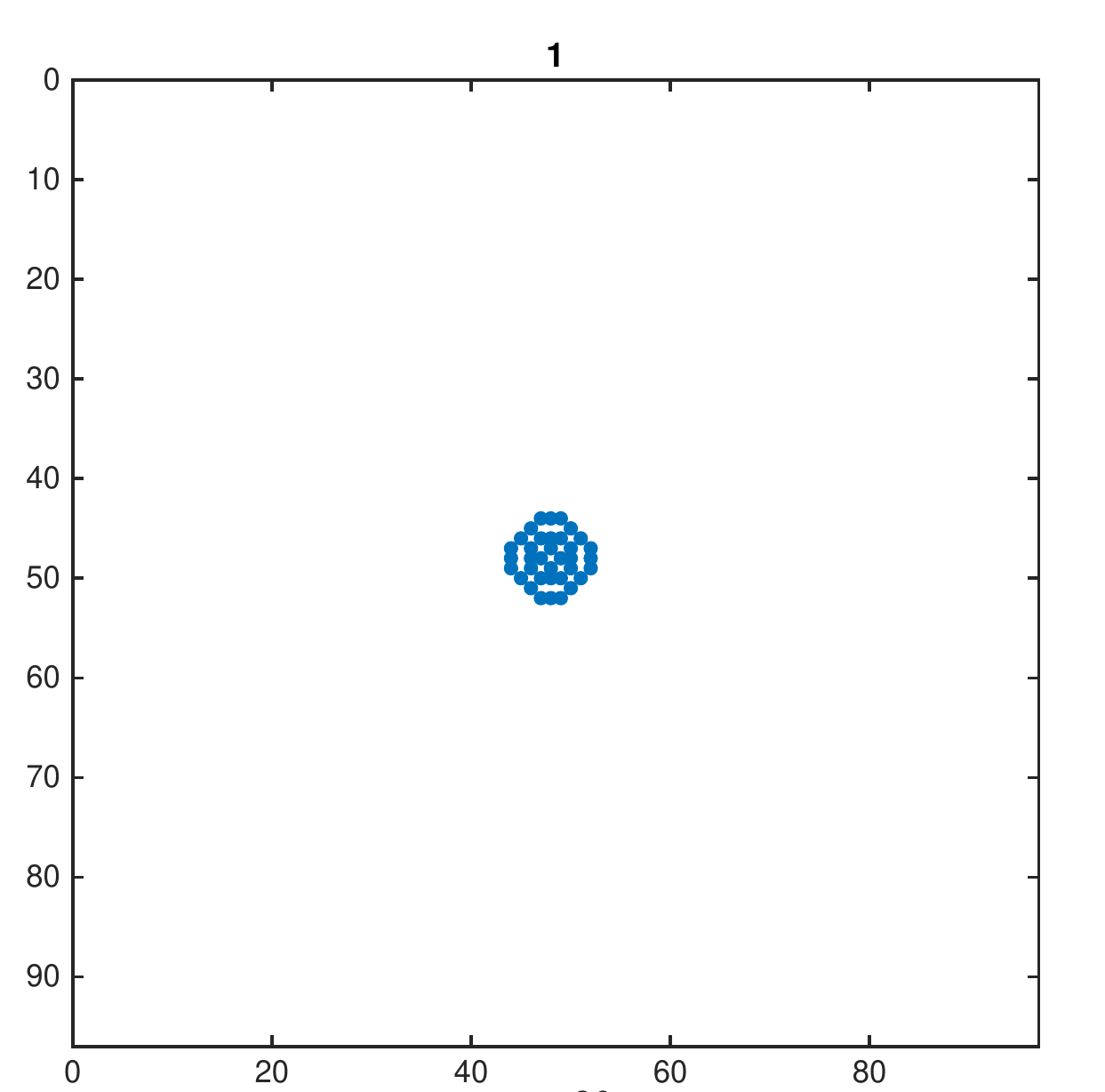} \quad
\includegraphics[scale=0.3]{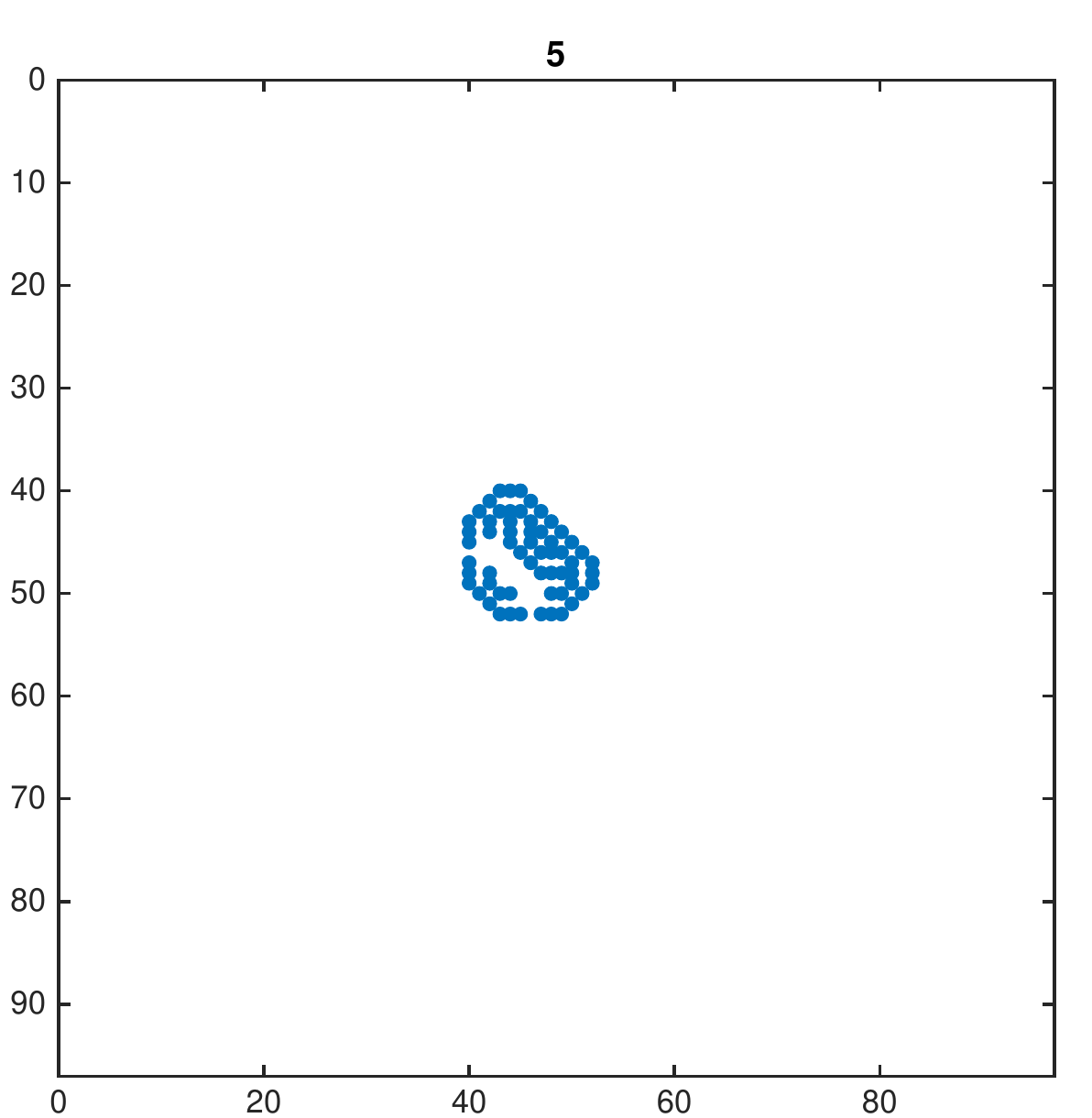} \quad
\includegraphics[scale=0.3]{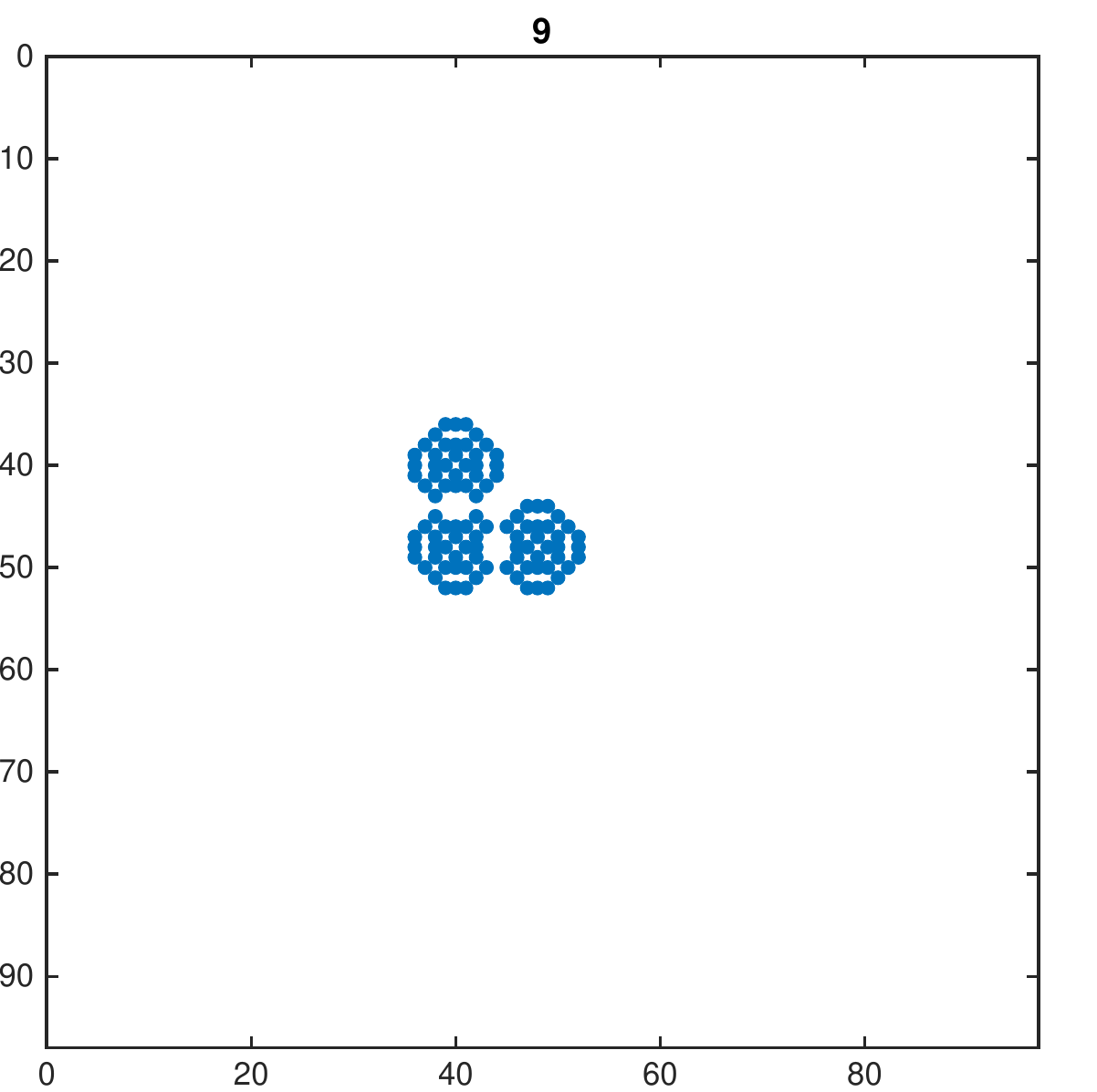} \\
\includegraphics[scale=0.3]{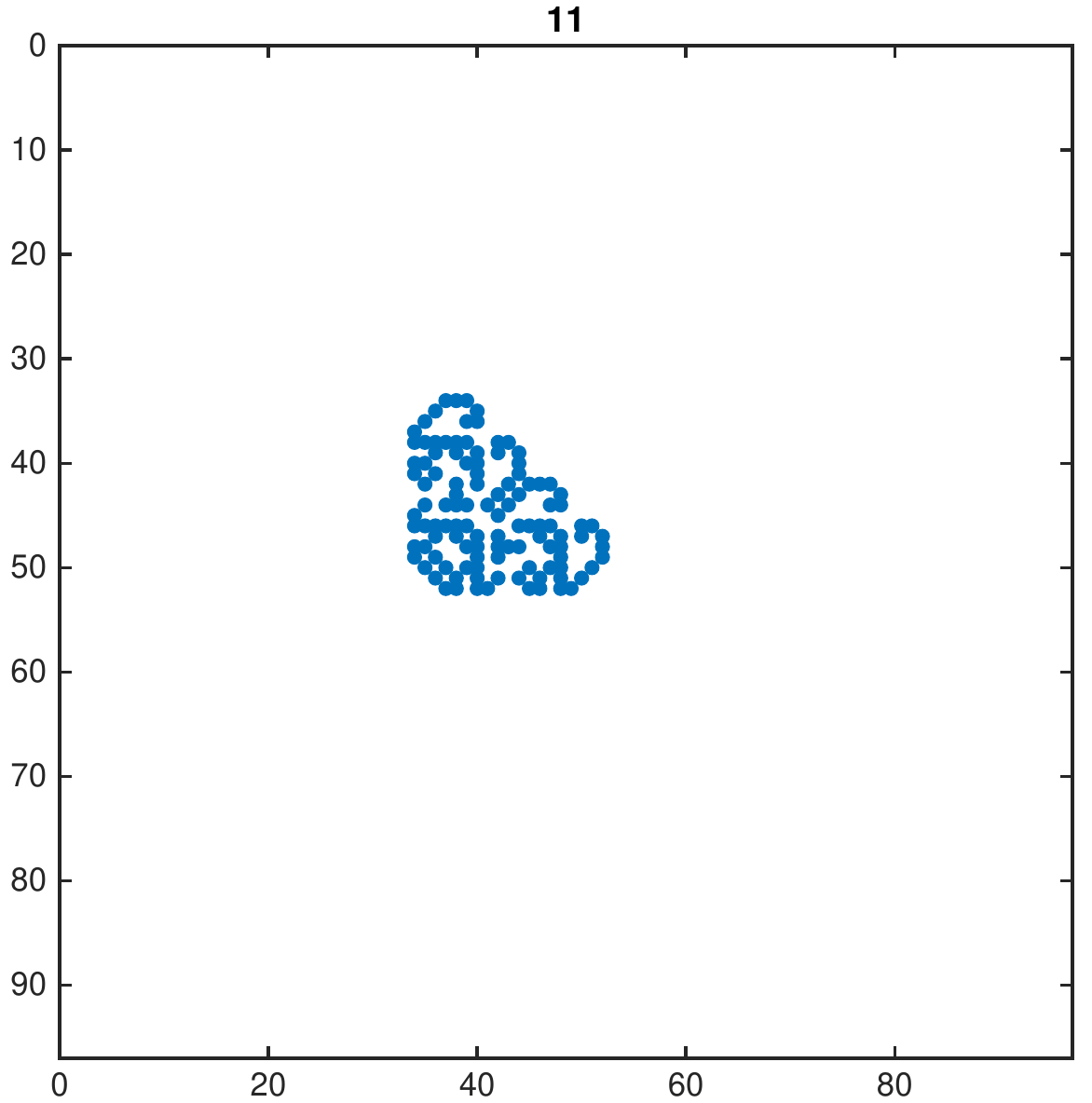} \quad
\includegraphics[scale=0.3]{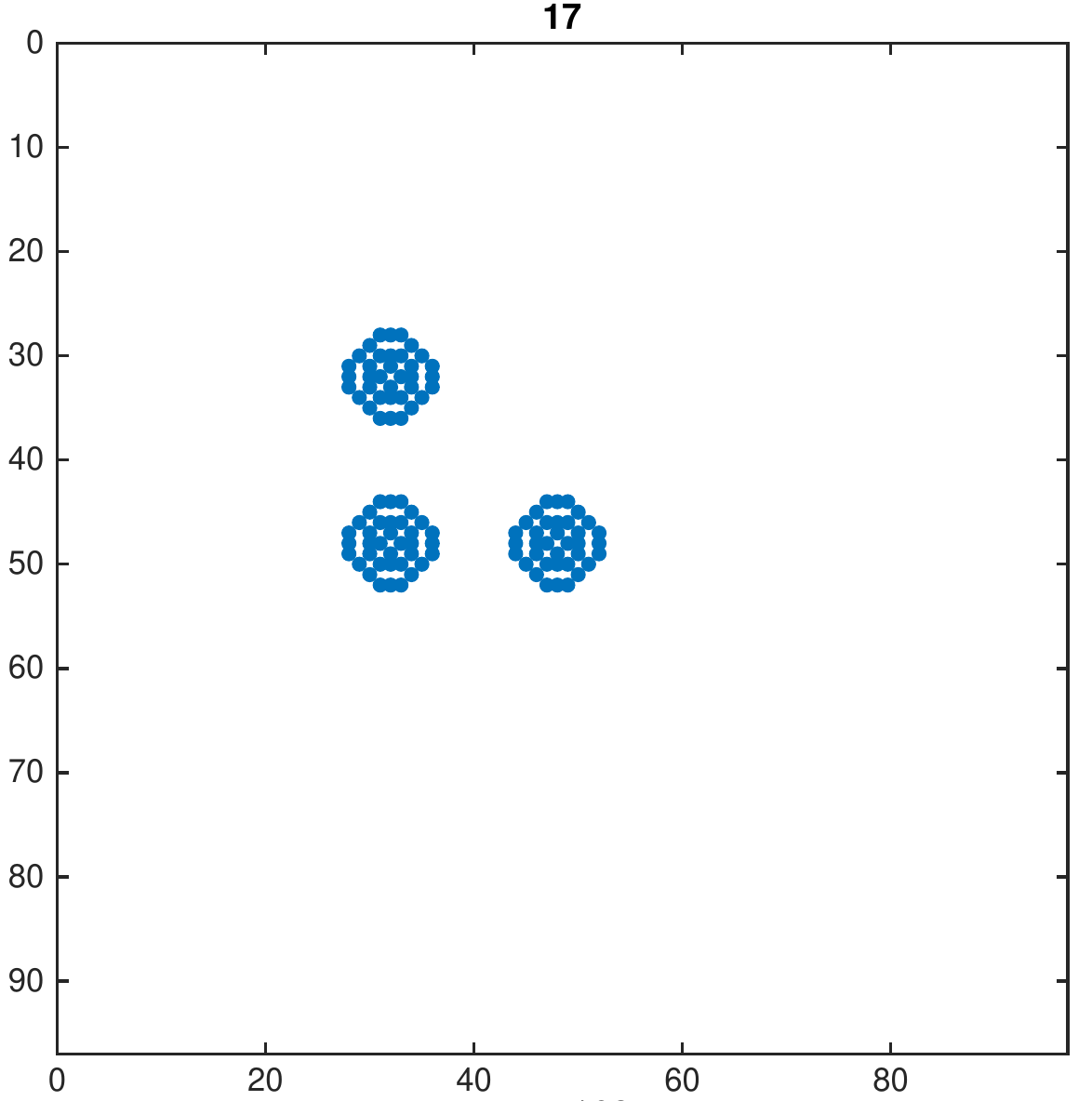} \quad
\includegraphics[scale=0.3]{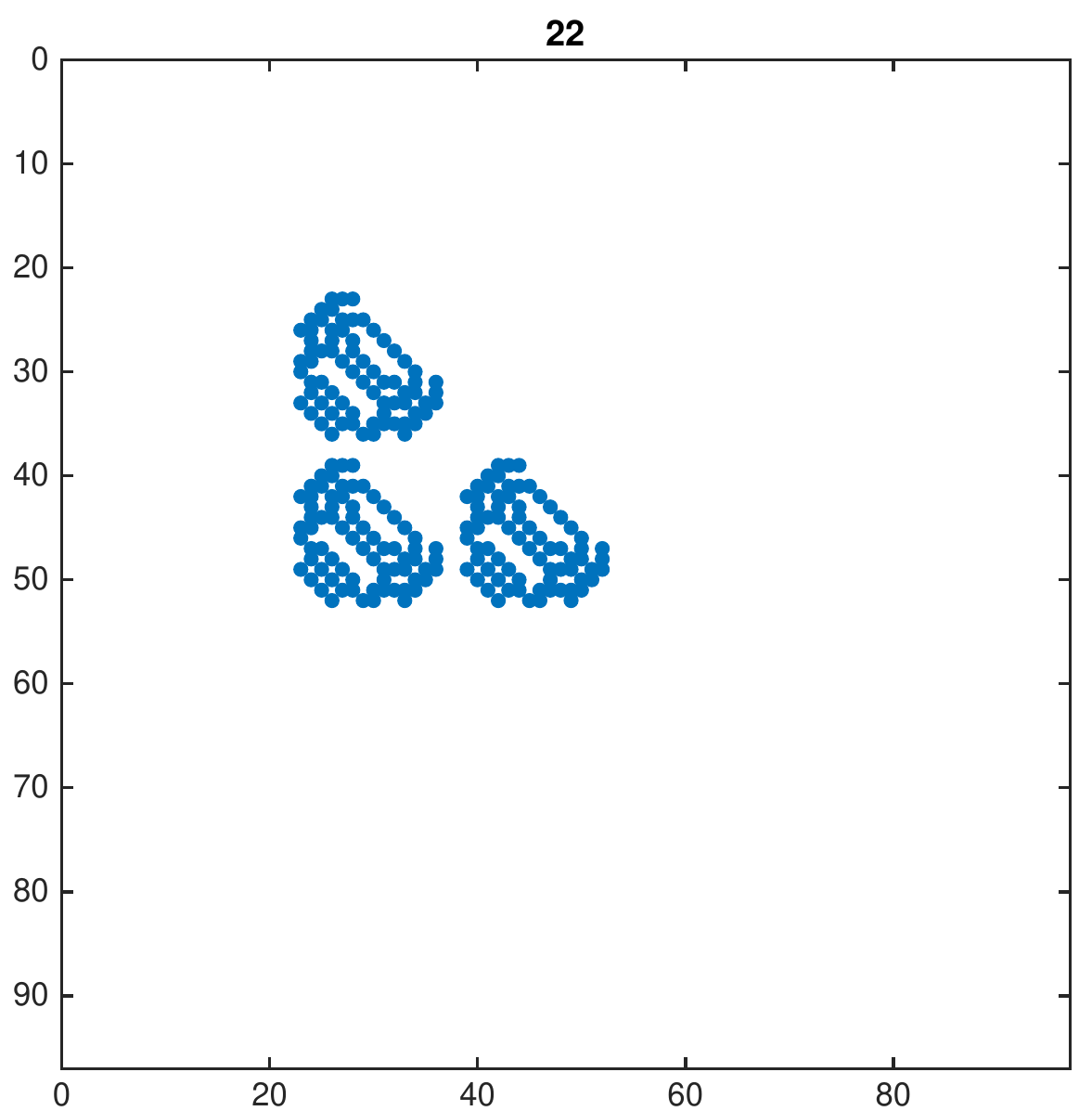} \\
\includegraphics[scale=0.3]{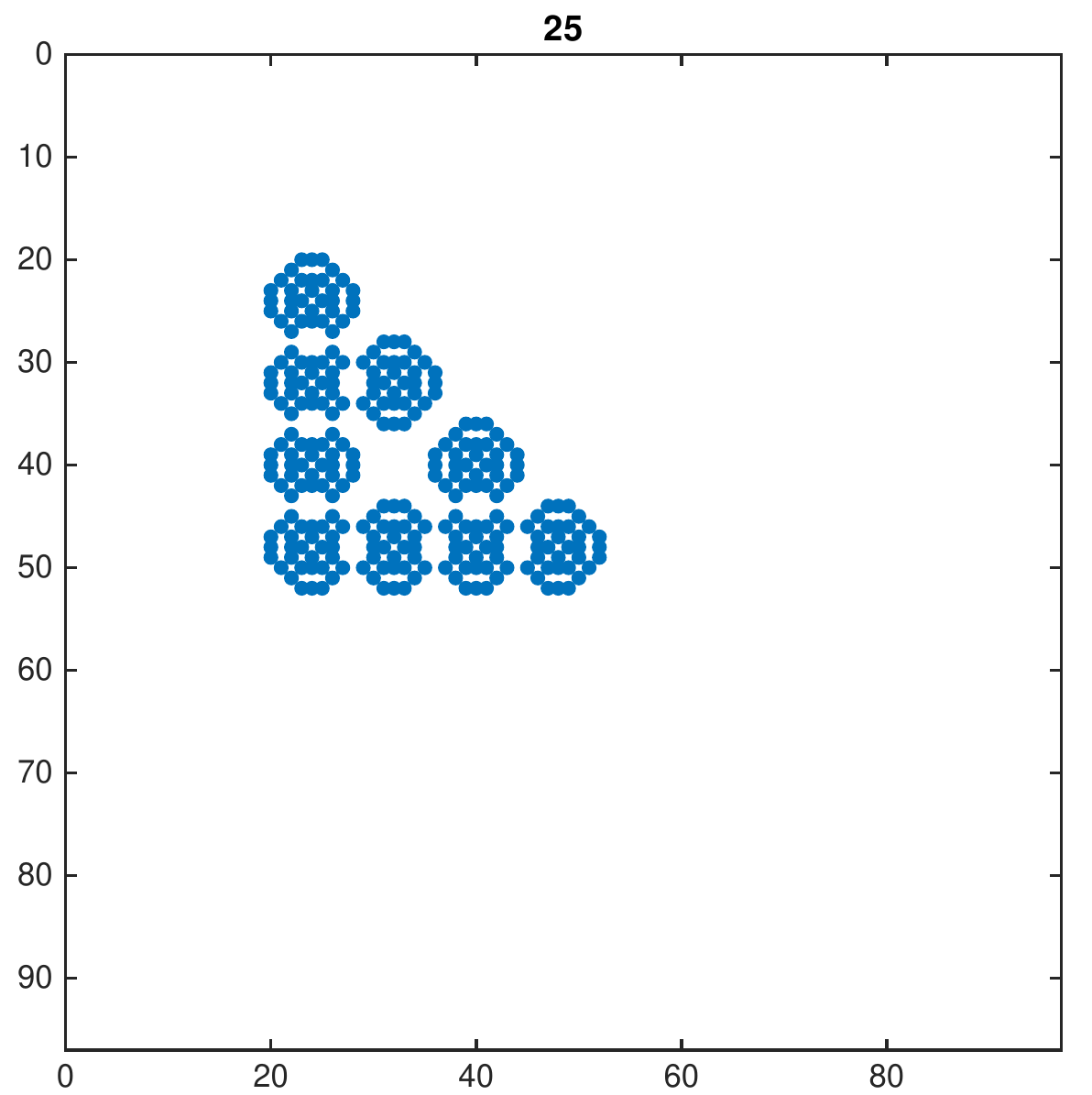} \quad
\includegraphics[scale=0.3]{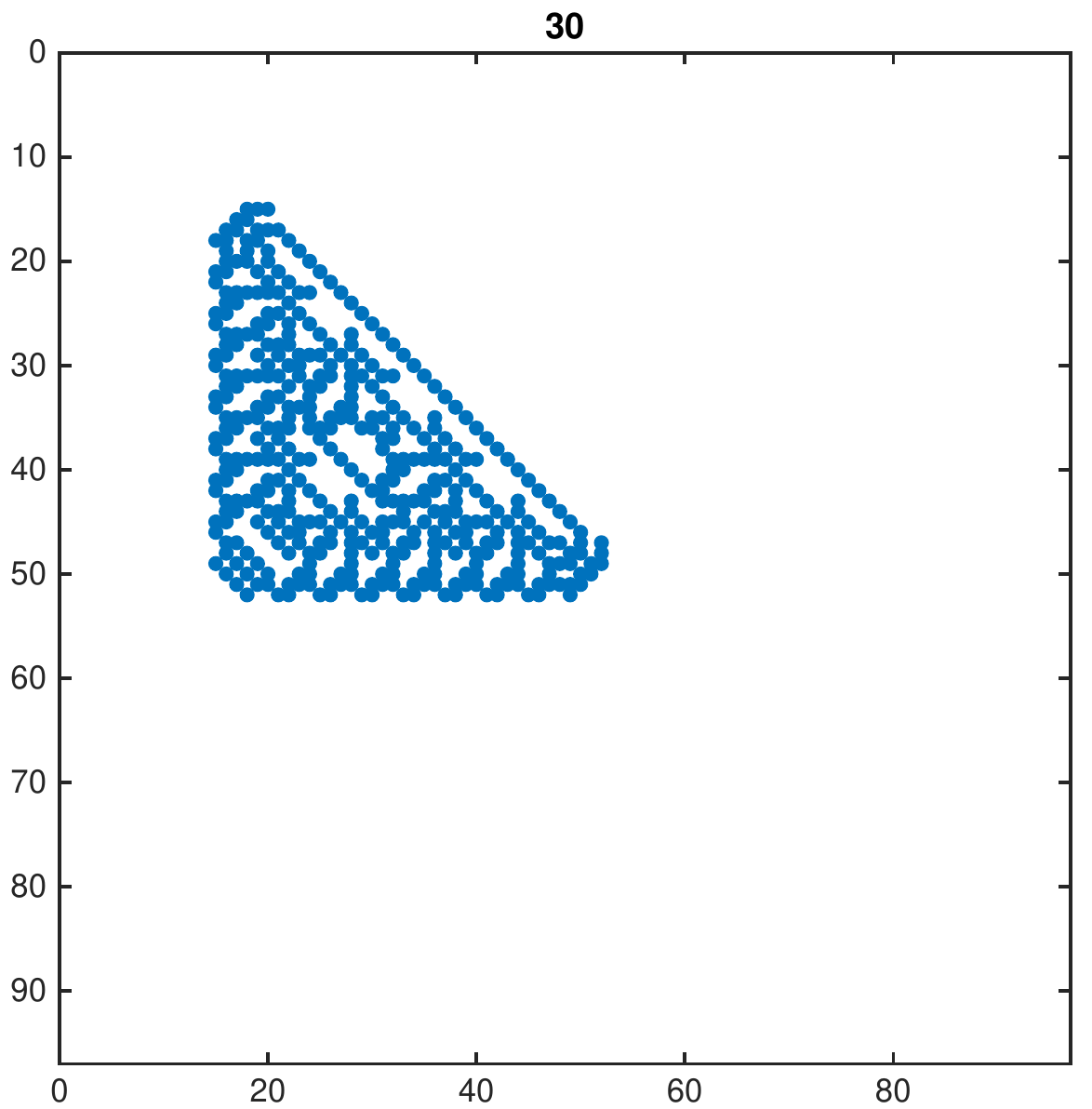} \quad
\includegraphics[scale=0.3]{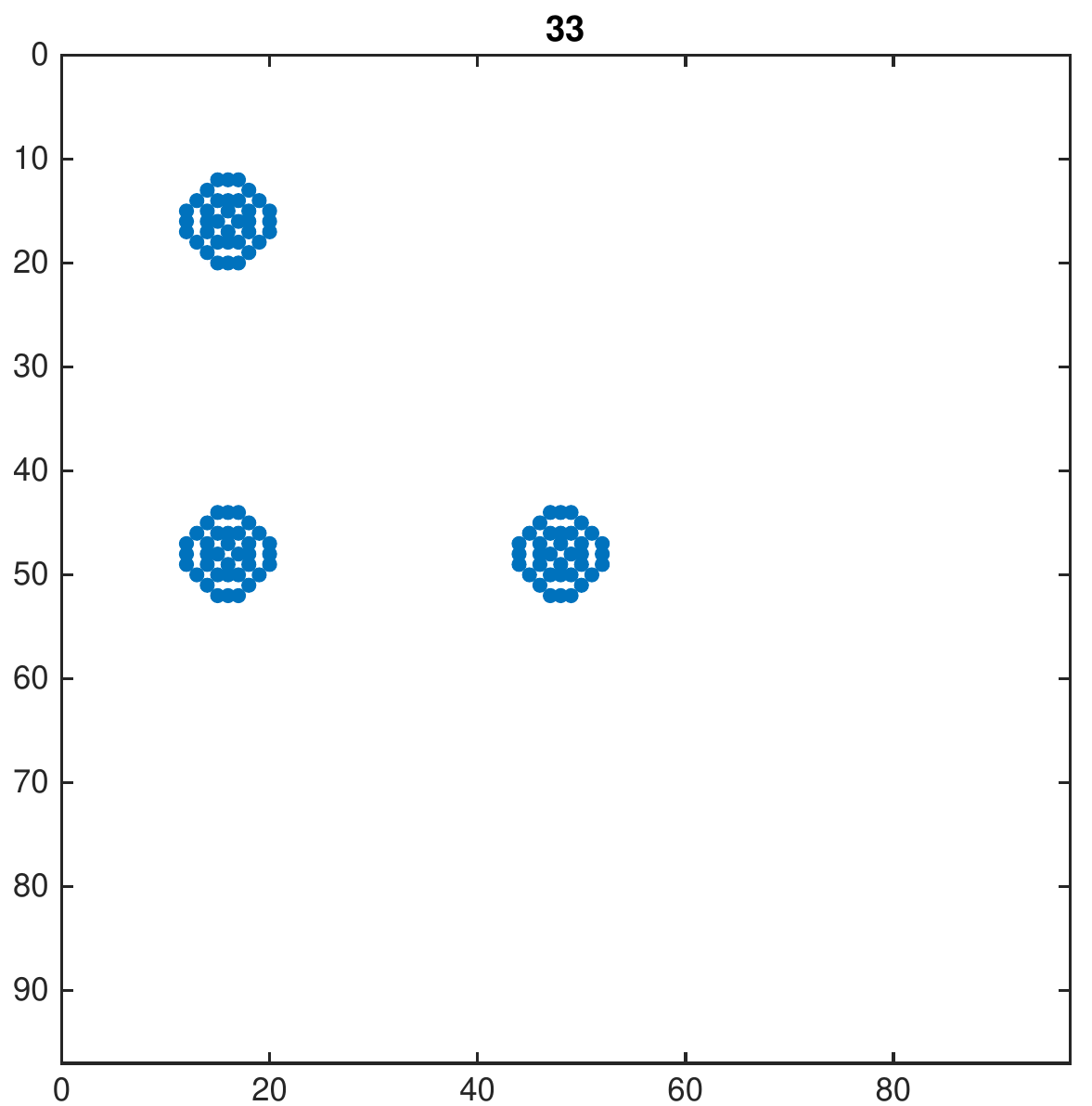}
\end{center}
\caption{A self-replicating pattern generation process for a linear cellular automaton addressed in Example \ref{eg:2d-self-replication}. It is seen that three copies of the initial pattern are reproduced at the 16th and 32nd iterations (the middle one and the right bottom corner). Some interesting transient patterns are also observed.}
\label{fig:2d-self-replication}
\end{figure}
\end{example}

\section{Mixing Property for Local Rules on Convex Hull}

In the previous section, Theorem \ref{thm:rectangle-k-mixing} and Remark \ref{rmk:rectangle-k-mixing-permutive} address that a corner permutive cellular automaton with local rule defined on a hypercuboid is $k$-mixing with respect to the uniform Bernoulli measure for $k \geq 1$. This section investigates a discrimination for determining whether or not a cellular automaton with local rule defined on a convex hull is mixing with respect to the uniform Bernoulli measure.

Suppose $F$ is a CA with the local rule $f$ defined on a $d$-dimensional convex hull $\mathcal{C}$ that is generated by the vertex set $C=\{\mathbf{v}_1, \mathbf{v}_2, \ldots, \mathbf{v}_{\ell}\}$, where $\mathbf{v}_i = (v_{i,1}, v_{i,2}, \ldots, v_{i,d})$ for $1 \leq i \leq \ell$. The mixing property of a CA is examined by the following algorithm.

\begin{algorithm}[Mixing Algorithm]
\mbox{}
\begin{enumerate}[\bf {MA}1.]
\item There exists $\mathbf{v}_n \in C$ such that $v_{n,j} > M_j(C \setminus \{\mathbf{v}_n\})$ for some $j$ and $v_{n,j} > 0$.
\item There exists $\mathbf{v}_n \in C$ such that $v_{n,j} < m_j(C \setminus \{\mathbf{v}_n\})$ for some $j$ and $v_{n,j} < 0$.
\item Suppose neither \textbf{MA1} nor \textbf{MA2} holds, and $\mathbf{v}_n \in C$ is the vertex such that $v_{n,j} = M_j(C \setminus \{\mathbf{v}_n\})$ and $v_{n,j} > 0$ or $v_{n,j} = m_j(C \setminus \{\mathbf{v}_n\})$ and $v_{n,j} < 0$ for some $j$. Let $C' = \{ \mathbf{v} \in C: \pi_j(\mathbf{v}) = v_{n, j}\}$ and let $C^1 \subset \mathbb{Z}^{d-1}$ be the collection of vertices obtained by removing the $j$th coordinate of elements in $C'$. Apply the Mixing Algorithm to $C^1$.
\end{enumerate}
\end{algorithm}

A set $C \in \mathbb{Z}^d$ is said to \emph{satisfy the Mixing Algorithm} (at $\mathbf{v}_n$) if (1) $C$ satisfies either \textbf{MA1} or \textbf{MA2}, (2) $C^1$ that is constructed in \textbf{MA3} satisfies either \textbf{MA1} or \textbf{MA2}, or (3) repeating the procedure in \textbf{MA3} to construct $C^2, C^3, \ldots$ so that $C^j$ satisfies either \textbf{MA1} or \textbf{MA2} for some $j$.

Theorem \ref{thm:polygon-k-mixing} extends Theorem \ref{thm:rectangle-k-mixing} and Remark \ref{rmk:rectangle-k-mixing-permutive} to the case that the local rule of a CA is defined on a multidimensional convex hull.

\begin{theorem}\label{thm:polygon-k-mixing}
Suppose $F$ is a CA with the local rule $f$ defined on $\mathcal{C}$, where $\mathcal{C}$ is a $d$-dimensional convex hull generated by the vertex set $C=\{\mathbf{v}_1, \mathbf{v}_2, \ldots, \mathbf{v}_{\ell}\}$. If $f$ is permutive at $\mathbf{v}_n \in C$ and $C$ satisfies the Mixing Algorithm at $\mathbf{v}_n$, then $F$ is $k$-mixing with respect to the uniform Bernoulli measure for all $k \geq 1$.
\end{theorem}

Before demonstrating the theorem, the following examples clarify the essential concepts of Theorem \ref{thm:polygon-k-mixing}.

\begin{example}\label{eg:2d-mixing-polygon-case1}
Suppose $\mathcal{A}=\{0,1,2,3\}$. Let $\mathcal{C}$ be the polygon in $2$-dimensional lattice generated by the set $C = \{\mathbf{v}_1, \mathbf{v}_2, \mathbf{v}_3, \mathbf{v}_4, \mathbf{v}_5\}$, where
$$
\mathbf{v}_1 = (-1,-1), \mathbf{v}_2 = (-1,1), \mathbf{v}_3 = (0,2), \mathbf{v}_4 = (1,1), \mathbf{v}_5 = (1,-1).
$$

Suppose the local rule $f: \mathcal{A}^{\mathcal{C}} \to \mathcal{A}$ is given by
\begin{align*}
f (x_{\mathcal{C}}) = 2(x_{-1,-1}+x_{-1,1}+x_{1,1}+x_{1,-1})+ 3 x_{0,2} \pmod 4.
\end{align*}
It is seen that $f$ is permutive in the variable $x_{0,2}$ since $f$ is linear and $\mathrm{gcd}(3, 4) = 1$. Moreover, $f$ satisfies \textbf{MA1} follows from $v_{3, 2}>v_{i, 2}$ for $i \in \{1,2,4,5\}$ and $v_{3, 2} > 0$. Theorem \ref{thm:polygon-k-mixing} indicates that the CA $F$ with the local rule $f$ is $k$-mixing with respect to the uniform Bernoulli measure for $k \in \mathbb{N}$.
\end{example}

\begin{example}\label{eg:2d-mixing-polygon-case2}
Suppose $\mathcal{A}, \mathcal{C}$, and $C$ are the same as considered in Example \ref{eg:2d-mixing-polygon-case1}. Let the local rule $f: \mathcal{A}^{\mathcal{C}} \to \mathcal{A}$ is given by
\begin{align*}
f (x_{\mathcal{C}}) = 2(x_{-1,-1} \cdot x_{0,2} + x_{1,1} + x_{1,-1}) + x_{-1,1} \pmod 4.
\end{align*}
Although $f$ is nonlinear, a straightforward verification derives that $f$ is permutive at $\mathbf{v}_2 = (-1, 1)$.

Notably, $f$ does not satisfy neither one of \textbf{MA1}, \textbf{MA2} in the Mixing Algorithm since $v_{2, 1} = v_{1, 1} < m_1(C \setminus \{\mathbf{v}_1, \mathbf{v}_2\})$. \textbf{MA3} suggests that $C^1 = \{v_1^1, v_2^1\}$ should be testified via the Mixing Algorithm to see that if the nonlinear CA $F$ with the local rule $f$ is mixing, where $v_1^1 = v_{1, 2} = -1$ and $v_2^1 = v_{2, 2} = 1$. It comes from $C^1$ satisfies \textbf{MA1} and Theorem \ref{thm:polygon-k-mixing} that $F$ is $k$-mixing with respect to the uniform Bernoulli measure for all $k \geq 1$.
\end{example}

Examples \ref{eg:2d-mixing-polygon-case1} and \ref{eg:2d-mixing-polygon-case2} are not difficult to verify the conditions requested in the Mixing Algorithm due to they are two-dimensional CAs.

\begin{example}\label{eg:3d-mixing-polygon}
Suppose $\mathcal{A}=\{0,1,2,3\}$. Let $\mathcal{C}$ be the convex hull in $3$-dimensional lattice generated by the set $C = \{\mathbf{v}_i\}_{i=1}^{12}$, where
\begin{align*}
\mathbf{v}_1 &= (0,2,-1), & \mathbf{v}_2 &= (-1,0,-1), & \mathbf{v}_3 &= (0,-2,-1), \\
\mathbf{v}_4 &=(2,-2,-1), & \mathbf{v}_5 &= (3,0,-1), & \mathbf{v}_6 &= (2,2,-1), \\
\mathbf{v}_7 &= (0,2,1), & \mathbf{v}_8 &=(-1,0,1), & \mathbf{v}_9 &= (0,-2,1), \\
\mathbf{v}_{10} &= (2,-2,1), & \mathbf{v}_{11} &= (3,0,1), & \mathbf{v}_{12} &= (2,2,1).
\end{align*}

Suppose the local rule $f: \mathcal{A}^{\mathcal{C}} \to \mathcal{A}$ is given by
\begin{align*}
f (x_{\mathcal{C}}) = &x_{\mathbf{v}_1} x_{\mathbf{v}_2} + 2(x_{\mathbf{v}_4}+x_{\mathbf{v}_5} + x_{\mathbf{v}_7}^2 \cdot x_{\mathbf{v}_8}) + 3 x_{\mathbf{v}_6} + \\
 &x_{\mathbf{v}_3} x_{\mathbf{v}_9} x_{\mathbf{v}_{10}} + 2 x_{\mathbf{v}_{11}}^2 x_{\mathbf{v}_{12}} \pmod 4.
\end{align*}
A careful elucidation deduces that $f$ is permutive in the variable $x_{2,2,-1}=x_{\mathbf{v}_6}$. Since $v_{i, 3} = m_3(C) = -1 < 0$ for $1 \leq i \leq 6$, $f$ does not belong to the first two criteria of the Mixing Algorithm. Let $C^1 = \{\mathbf{v}_i^1\}_{i=1}^6$, as indicated in \textbf{MA1}, where
\begin{align*}
\mathbf{v}_1^1 &= (0,2), &\mathbf{v}_2^1 &= (-1,0), &\mathbf{v}_3^1 &= (0,-2), \\
\mathbf{v}_4^1 &= (2,-2), &\mathbf{v}_5^1 &= (3,0), &\mathbf{v}_6^1 &= (2,2).
\end{align*}
It remains to verify whether or not $C^1$ satisfies the Mixing Algorithm.

Since $v_{6, 2}^1 = M_2(C^1) > 0$ and $v_{1, 2}^1 = v^1_{6, 2}$, $C^1$ does not satisfy \textbf{MA1} and \textbf{MA2}. It follows that $C^2 = \{\mathbf{v}^2_1, \mathbf{v}^2_6\}$ is constructed, where $\mathbf{v}^2_1 = 0, \mathbf{v}^2_6 = 2$. The fact that $C^2$ satisfies \textbf{MA1} demonstrates the $k$-mixing property of $F$ with respect to the uniform Bernoulli measure for all $k \in \mathbb{N}$ via Theorem \ref{thm:polygon-k-mixing}.

This completes the illustration of Example \ref{eg:3d-mixing-polygon}.
\end{example}

Theorem \ref{thm:polygon-k-mixing} is demonstrated via an analogous argument to the proof of Theorem \ref{thm:rectangle-k-mixing} with a little modification. To make the present paper more compact, the following demonstration addresses the key idea rather than the detailed discussion.

\begin{proof}[Proof of Theorem \ref{thm:polygon-k-mixing}]
To sketch the key idea of the proof of Theorem \ref{thm:polygon-k-mixing}, it suffices to concentrate on the two-dimensional case. The investigation of the $k$-mixing property of multidimensional CAs can be completed via similar deliberation iteratively.

Theorem \ref{thm:polygon-k-mixing} and the Mixing Algorithm for the case $d = 2$ is presented as follows. Let $f: \mathcal{A}^{\mathcal{C}} \to \mathcal{A}$ be a local rule defined on a polygon $\mathcal{C}$ generated by $C = \{\mathbf{v}_i = (v_{i, 1}, v_{i, 2})\}_{i=1}^{\ell} \subset \mathbb{Z}^2$. Suppose $f$ is permutive at $\mathbf{v}_n = (v_{n, 1}, v_{n, 2})$ and $\mathbf{v}_n$ satisfies one of the following conditions:
\begin{enumerate}[\bf (I)]
\item $v_{n,j} > v_{i, j}$ (respectively $v_{n, j} < v_{i, j}$) for all $i \neq n$ and $v_{n, j}>0$ (respectively $v_{n, j}<0$) for some $j \in \{1, 2\}$.
\item There exists $j \in \{1, 2\}$ such that $v_{n, j} = M_j(C \setminus \{\mathbf{v}_n\})$ (respectively $v_{n, j} = m_j(C \setminus \{\mathbf{v}_n\})$) and $v_{n, j}>0$ (respectively $v_{n, j}<0$). Let
$$
C' = \{\mathbf{u} = (u_1, u_2) \in C: u_j = \mathbf{v}_{n, j}\}
$$
and let $\overline{j}= 3 - j$. Then $v_{n, \overline{j}} > M_{\overline{j}} (C' \setminus \{\mathbf{v}_n\})$ (respectively $v_{n, \overline{j}} < m_{\overline{j}} (C' \setminus \{\mathbf{v}_n\})$) and $v_{n, \overline{j}} > 0$ (respectively $v_{n, \overline{j}} < 0$).
\end{enumerate}
Then a CA $F$ with the local rule $f$ is $k$-mixing with respect to the uniform Bernoulli measure for all positive integer $k$.

Suppose $\mathbf{v}_n$ satisfies Condition (\textbf{I}). Then the coordinates of $\mathbf{v}_n$ satisfy one of the following conditions specifically.
\begin{enumerate}[\bf ({I.}a)]
\item $v_{n, 1} > v_{i, 1}$ for $i \neq n$, and $v_{i, 2} < v_{n, 2} < v_{j, 2}$ for some $i, j \neq n$.
\item $v_{n, 1} < v_{i, 1}$ for $i \neq n$, and $v_{i, 2} < v_{n, 2} < v_{j, 2}$ for some $i, j \neq n$.
\item $v_{n, 2} > v_{i, 2}$ for $i \neq n$, and $v_{i, 1} < v_{n, 1} < v_{j, 1}$ for some $i, j \neq n$.
\item $v_{n, 2} < v_{i, 2}$ for $i \neq n$, and $v_{i, 1} < v_{n, 1} < v_{j, 1}$ for some $i, j \neq n$.
\item $v_{n, 1} > v_{i, 1}$ for $i \neq n$, and $v_{n, 2} \geq M_2(C)$ or $v_{n, 2} \leq M_2(C)$.
\item $v_{n, 1} < v_{i, 1}$ for $i \neq n$, and $v_{n, 2} \geq M_2(C)$ or $v_{n, 2} \leq M_2(C)$.
\item $v_{n, 2} > v_{i, 2}$ for $i \neq n$, and $v_{n, 1} \geq M_1(C)$ or $v_{n, 1} \leq M_1(C)$.
\item $v_{n, 2} < v_{i, 2}$ for $i \neq n$, and $v_{n, 1} \geq M_1(C)$ or $v_{n, 1} \leq M_1(C)$.
\end{enumerate}
The demonstration of $\mathbf{v}_n$ satisfying \textbf{(I.a)} is addressed. The other cases can be done similarly. Given two cylinders $\mathbf{C}_0 = <(\mathbf{v}_1^0, c_1^0), \ldots, (\mathbf{v}_{l_0}^0, c_{l_0}^0)>$ and $\mathbf{C}_1 = <(\mathbf{v}_1^1, c_1^1), \ldots, (\mathbf{v}_{l_1}^1, c_{l_1}^1)>$ in $\mathbf{X}$, similar discussion to the proof of Theorem \ref{thm:rectangle-k-mixing} would show that $\mu(\mathbf{C}_0 \cap F^{-n} \mathbf{C}_1) = \mu(\mathbf{C}_0) \mu(\mathbf{C}_1)$ for $n$ large enough. It is worth emphasizing that, to choose $n$ properly, the following specific procedure during the evaluation of $\mu(\mathbf{C}_0 \cap F^{-n} \mathbf{C}_1)$ is essential: For $q_1, q_2 \in \{1, 2, \ldots, l_1\}$, the $n$th preimage $F^{-n}(<\mathbf{v}^1_{q_1}, c^1_{q_1}>)$ of the cylinder $<\mathbf{v}^1_{q_1}, c^1_{q_1}>$ has to be considered before $F^{-n}(<\mathbf{v}^1_{q_2}, c^1_{q_2}>)$ if and only if
\begin{enumerate}[\bf 1)]
  \item $v^1_{q_1, 1} < v^1_{q_2, 1}$;
  \item $v^1_{q_1, 1} = v^1_{q_2, 1}$ and $v^1_{q_1, 2} < v^1_{q_2, 2}$.
\end{enumerate}
With the notion of proper order for the computation of $F^{-n} \mathbf{C}_1)$, an analogous investigation to the proof of Theorem \ref{thm:rectangle-k-mixing} reaches the desired result, i.e., $F$ is strongly mixing with respect to the uniform Bernoulli measure $\mu$. Moreover, it can be verified that $F$ is $k$-mixing with respect to the uniform Bernoulli measure for $k \geq 1$.

Suppose $\mathbf{v}_n$ satisfies Condition (\textbf{II}). It is seen without difficulty that the coordinates of $\mathbf{v}_n$ can be described as the following cases.
\begin{enumerate}[\bf ({II.}a)]
\item $v_{n, 1} = M_1(C)$ and $v_{n, 2} = M_2(C)$.
\item $v_{n, 1} = M_1(C)$ and $v_{n, 2} = m_2(C)$.
\item $v_{n, 1} = m_1(C)$ and $v_{n, 2} = M_2(C)$.
\item $v_{n, 1} = m_1(C)$ and $v_{n, 2} = m_2(C)$.
\item $v_{n, 1} \in \{M_1(C), m_1(C)\}$ but $v_{n, 2} \notin \in \{M_2(C), m_2(C)\}$.
\item $v_{n, 2} \in \{M_2(C), m_2(C)\}$ but $v_{n, 1} \notin \in \{M_1(C), m_1(C)\}$.
\end{enumerate}
Cases \textbf{(II.e)} and \textbf{(II.f)} can be verified via the discussion above, it remains to study the other four cases. Notably, for cases \textbf{(II.a)} to \textbf{(II.d)}, $\mathcal{C}$ can be embedded into a rectangle $\overline{\mathcal{C}}$ so that there exists a unique $\overline{f}: \mathcal{A}^{\overline{\mathcal{C}}} \to \mathcal{A}$ with $\overline{f}|_{\mathcal{A}^{\mathcal{C}}} = f$ and $\overline{f}|_{\mathcal{A}^{\overline{\mathcal{C}} \setminus \mathcal{C}}} = 0$. More specifically, $\overline{F} = F$ and $\overline{f}$ is corner permutive, where $\overline{F}$ is a CA with the local rule $\overline{f}$. Theorem \ref{thm:rectangle-k-mixing} indicates that $\overline{F}$ is $k$-mixing with respect to the uniform Bernoulli measure for $k \geq 1$, and so is $F$.

This completes the proof.
\end{proof}

\section{Conclusion and Discussion}

This elucidation investigates sufficient conditions for the strongly mixing property of a multidimensional cellular automaton $F$ with the local rule $f$ defined on a bounded region $D \subset \mathbb{Z}^d$. Theorem \ref{thm:rectangle-k-mixing} reveals that, when $D$ is a $d$-dimensional hypercuboid and $f$ is corner permutive, $F$ is $k$-mixing with respect to the uniform Bernoulli measure for $k \geq 1$.

Observe that a hypercuboid $D$ is a convex hull generated by its apexes. More precisely, there exists $\{k_i, K_i\}_{i=1}^d \subset \mathbb{Z}$ and
$$
C = \{\mathbf{v} = (v_1, \ldots, v_d): v_i \in \{k_i, K_i\} \text{ fo all } i\}
$$
such that
$$
D = \mathrm{poly}(C) = \{\Sigma a_j \mathbf{v}_j: 0 \leq a_j \leq 1, \Sigma a_j = 1, \mathbf{v}_j \in C\} \cap \mathbb{Z}^d.
$$
The assumption that $f$ is corner permutive is $f$ is a permutation at $\mathbf{v}$ for some $\mathbf{v} \in C$.

Theorem \ref{thm:polygon-k-mixing} is an extension of above observation, which addresses that, if $D$ is a multidimensional convex hull generated by a ``minimal" vertex set $C$ and $f$ is a permutation at $\mathbf{v}$ for some $\mathbf{v} \in C$, then $F$ is $k$-mixing with respect to the uniform Bernoulli measure for all positive integer $k$. Herein a generating set $C$ is called \emph{minimal} if $C \subseteq C'$ for all $C'$ such that $\mathrm{poly}(C') = \mathrm{poly}(C)$.

It is natural to ask that is there any possibility to weaken the sufficient conditions proposed in Theorems \ref{thm:rectangle-k-mixing} and \ref{thm:polygon-k-mixing}? The upcoming example infers a non-corner-permutive cellular automaton $F$ which is even nonergodic.

\begin{example}\label{eg:not-corner-permutive-not-mixing}
Suppose $\mathcal{A} = \{0, 1, 2, 3\}$ and a local rule $f$ given by
\begin{align*}
f\left(
  \begin{array}{ccc}
  x_{0,2} & x_{1,2} & x_{2,2} \\
  x_{0,1} & x_{1,1} & x_{2,1} \\
  x_{0,0} & x_{1,0} & x_{2,0} \\
  \end{array}
\right)
=2x_{0,1} + x_{1,1} + \left\lfloor\dfrac{2x_{22}}{3}\right\rfloor + 3 x_{2,2} + 2x_{0,0} \cdot x_{2,0}
\end{align*}
is defined on a polygon, where $\lfloor q \rfloor$ refers to the greatest integer which is less than or equal to $q$. A straightforward examination demonstrates that $f$ is permutive in the variable $x_{1,1}$ and is not corner permutive. Furthermore, it is seen that $F^{-1}(\{<(0,0),i>, <(0,1),2>\}) = \varnothing$ for all $i \in \mathcal{A}$, where $F$ is the CA with the local rule $f$. This makes $F$ not ergodic, and thus not $k$-mixing.
\end{example}

\begin{remark}
It is remarkable that the result in the present investigation can be extended to any Markov measure $\nu$ such that the cellular automaton $F$ is $\nu$-invariant. The discussion is similar but more complicated.
\end{remark}

The elucidation of $k$-mixing property of a cellular automaton can be applied to the study of ergodicity of a multidimensional cellular automaton, which is covered in the further paper.

\bibliographystyle{amsplain} 
\bibliography{../../grece}

\providecommand{\bysame}{\leavevmode\hbox to3em{\hrulefill}\thinspace}
\providecommand{\MR}{\relax\ifhmode\unskip\space\fi MR }
\providecommand{\MRhref}[2]{%
  \href{http://www.ams.org/mathscinet-getitem?mr=#1}{#2}
}
\providecommand{\href}[2]{#2}
\begin{thebibliography}{10}

\bibitem{ADF-IPL2013}
L.~Acerbia, A.~Dennunziob, and E.~Formenti, \emph{Surjective multidimensional
  cellular automata are non-wandering: {A} combinatorial proof}, Inform.
  Process. Lett. \textbf{113} (2013), 156--159.

\bibitem{Akin-GM2010}
H.~Ak{\i}n, \emph{On strong mixing property of cellular automata with respect
  to {Markov} measures}, Gen. Math. \textbf{18} (2010), 19--30.

\bibitem{AP-JCSS1972}
S.~Amoroso and Y.~N. Patt, \emph{Decision procedures for surjectivity and
  injectivity of parallelmaps for tessellation structures}, J. Comput. System
  Sci. \textbf{6} (1972), 448--464.

\bibitem{Bennett-IJRD1973}
C.~H. Bennett, \emph{Logical reversibility of computation}, IBM J Res. Develop.
  \textbf{17} (1973), 525--532.

\bibitem{BKM-PD1997}
F.~Blanchard, P.~Kurka, and A.~Maass, \emph{Topological and measure-theoretic
  properties of one-dimensional cellular automata}, Phys. D \textbf{103}
  (1997), 86--99.

\bibitem{BT-AIHPPS2000}
F.~Blanchard and P.~Tisseur, \emph{Some properties of cellular automata with
  equicontinuity points}, Ann. Inst. Henri Poincar\'{e}, Probabilit\'{e}
  Satistiques \textbf{36} (2000), 569--582.

\bibitem{CFMM-TCS2000}
G.~Cattaneo, E.~Formenti, G.~Manzini, and L.~Margara, \emph{Ergodicity,
  transitivity, and regularity for linear cellular automata over $z_m$},
  Theoret. Comput. Sci. \textbf{233} (2000), 147--164.

\bibitem{CFM+-TCS1999}
G.~Cattaneo, E.~Formenti, L.~Margara, and G.~Mauri, \emph{On the dynamical
  behavior of chaotic cellular automata}, Theoret. Comput. Sci. \textbf{217}
  (1999), 31--51.

\bibitem{CA-2014}
C.-H. Chang and H.~Ak{\i}n, \emph{Some ergodic properties of invertible
  cellular automata}, 2014, Submitted.

\bibitem{DFV-2003}
B.~Durand, E.~Formenti, and G.~Varouchas, \emph{On undecidability of
  equicontinuity classification for cellular automata}, Discrete Mathematics
  and Theoretical Computer Science Proceedings AB (M.~Morvan and E.~Remila,
  eds.), 2003, pp.~117--128.

\bibitem{FD-FI2008}
F.~Farina and A.~Dennunzio, \emph{A predator-prey cellular automaton with
  parasitic interactions and environmental effects}, Fund. Inform. \textbf{83}
  (2008), 337--353.

\bibitem{Hed-MST1969}
G.~A. Hedlund, \emph{Endomorphisms and automorphisms of full shift dynamical
  system}, Math. Systems Theory \textbf{3} (1969), 320--375.

\bibitem{ION-JCSS1983}
M.~Ito, N.~Osato, and M.~Nasu, \emph{Linear cellular automata over
  $\mathbb{Z}_m$}, J. Comput. System Sci. \textbf{27} (1983), 125--140.

\bibitem{Kari-JCSS1994}
J.~Kari, \emph{Reversibility and surjectivity problems of cellular automata},
  J. Comput. System Sci. \textbf{48} (1994), 149--182.

\bibitem{Kari-TCS2005}
\bysame, \emph{Theory of cellular automata: {A} survey}, Theoret. Comput. Sci.
  \textbf{334} (2005), 3--33.

\bibitem{KO-2008}
J.~Kari and N.~Ollinger, \emph{Periodicity and immortality in reversible
  computing}, Lecture Notes in Computer Science, vol. 5162, pp.~419--430,
  Springer, 2008.

\bibitem{Kle-PAMS1997}
R.~Kleveland, \emph{Mixing properties of one-dimensional cellular automata},
  Proc. Amer. Math. Soc. \textbf{125} (1997), 1755--1766.

\bibitem{Lee-2009}
C.-L. Lee, \emph{Mixing property for multi-dimensional cellular automata},
  Master's thesis, National Chiao Tung University, Taiwan, Republic of China,
  2009.

\bibitem{LM-TCS2010}
P.~Di Lena and L.~Margara, \emph{On the undecidability of the limit behavior of
  cellular automata}, Theoret. Comput. Sci. \textbf{411} (2010), 1075--1084.

\bibitem{Lukkarila-JCA2010}
V.~Lukkarila, \emph{Sensitivity and topological mixing are undecidable for
  reversible one-dimensional cellular automata}, J. Cell. Autom. \textbf{5}
  (2010), 241--272.

\bibitem{MM-BR2003}
A.~Maass and S.~Martinez, \emph{Evolution of probability measures by cellular
  automata on algebraic topological {Markov} chains}, Biol Res. \textbf{35}
  (2003), 113--118.

\bibitem{Morit-IPL1992}
K.~Morita, \emph{Computation universality of one-dimensional reversible
  cellular automata}, Inform. Process. Lett. \textbf{42} (1992), 325--329.

\bibitem{Morit-JIPSJ1994}
\bysame, \emph{Reversible cellular automata}, J. Inform. Process. Soc. Jap.
  \textbf{35} (1994), 315--321.

\bibitem{Morit-2012}
\bysame, \emph{Reversible cellular automata}, Handbook of Natural Computing,
  Springer-Verlag Berlin Heidelberg, 2012, pp.~231--257.

\bibitem{MH-IT1989}
K.~Morita and M.~Harao, \emph{Computation universality of 1 dimensional
  reversible (injective) cellular automata}, IEICE Trans. \textbf{E72} (1989),
  758--762.

\bibitem{PY-ETDS2002}
M.~Pivato and R.~Yassawi, \emph{Limit measures for affine cellular automata},
  Ergodic Theory Dynam. Systems \textbf{22} (2002), 1269--1287.

\bibitem{She-MM1992}
M.~A. Shereshevsky, \emph{Ergodic properties of certain surjective cellular
  automata}, Monatsh. Math. \textbf{114} (1992), 305--316.

\bibitem{SR-CMP1991}
M.~Shirvani and T.~D. Rogers, \emph{On ergodic one-dimensional cellular
  automata}, Commun. Math. Phys. \textbf{136} (1991), 599--605.

\bibitem{Ula-PICoM1952}
S.~Ulam, \emph{Random process and transformations}, Proc. Int. Congress of
  Math. \textbf{2} (1952), 264--275.

\bibitem{vNeu-1966}
J.~von Neumann, \emph{Theory of self-reproducing automata}, Univ. of Illinois
  Press, Urbana, 1966.

\bibitem{Wal-1982}
P.~Walters, \emph{An introduction to ergodic theory}, Springer-Verlag New York,
  1982.

\bibitem{Willson-MST1975}
S.~Willson, \emph{On the ergodic theory of cellular automata}, Math. Systems
  Theory \textbf{9} (1975), 132--141.

\end{thebibliography}

\end{document}